\documentclass[a4paper,UKenglish,cleveref, autoref, thm-restate]{lipics-v2021}

\allowdisplaybreaks

\usepackage{xspace}

\usepackage{thm-restate}

\graphicspath{{./figures/}}

\newcommand{\etal}{\emph{et~al.}\xspace}
\newcommand{\A}{A}

\title{The Tight Spanning Ratio of the Rectangle Delaunay Triangulation} 

\titlerunning{The Tight Spanning Ratio of the Rectangle Delaunay Triangulation} 

\author{Andr\'e van Renssen}{University of Sydney, Australia}{andre.vanrenssen@sydney.edu.au}{0000-0002-9294-9947}{}
\author{Yuan Sha}{University of Sydney, Australia}{ysha3185@sydney.edu.au}{}{}

\author{Yucheng Sun}{ETH Zurich, Switzerland}{yucsun@ethz.ch}{}{}

\author{Sampson Wong}{BARC, University of Copenhagen, Denmark}{sawo@di.ku.dk}{}{}

\authorrunning{A. van Renssen, Y. Sha, Y. Sun and S. Wong} 

\Copyright{Andr\'e van Renssen and Yuan Sha and Yucheng Sun and Sampson Wong} 



\category{} 

\relatedversion{https://arxiv.org/abs/2211.11987} 

\nolinenumbers 

\EventEditors{John Q. Open and Joan R. Access}
\EventNoEds{2}
\EventLongTitle{42nd Conference on Very Important Topics (CVIT 2016)}
\EventShortTitle{CVIT 2016}
\EventAcronym{CVIT}
\EventYear{2016}
\EventDate{December 24--27, 2016}
\EventLocation{Little Whinging, United Kingdom}
\EventLogo{}
\SeriesVolume{42}
\ArticleNo{23}

\begin{document}


\maketitle

\begin{abstract}
Spanner construction is a well-studied problem and Delaunay triangulations are among the most popular spanners. Tight bounds are known if the Delaunay triangulation is constructed using an equilateral triangle, a square, or a regular hexagon. However, all other shapes have remained elusive. In this paper, we extend the restricted class of spanners for which tight bounds are known. We prove that Delaunay triangulations constructed using rectangles with aspect ratio $\A$ have spanning ratio at most $\sqrt{2} \sqrt{1+\A^2 + \A \sqrt{\A^2 + 1}}$, which matches the known lower bound.
\end{abstract}

\newpage

\section{Introduction}
\label{sec:introduction}
A geometric graph is a weighted graph in the plane where every vertex $v$ has coordinates $(x_v, y_v)$ and the weight of an edge between any two vertices is the Euclidean distance between its endpoints. A \emph{geometric spanner} is defined to be a class of subgraphs where the shortest path distance between any two vertices is at most the Euclidean distance between these two vertices multiplied by a constant $t$. The smallest constant $t$ for which this property holds is called the \emph{spanning ratio} or \emph{stretch factor} of the geometric spanner. A comprehensive overview on the topic of geometric spanners can be found in the book by Narasimhan and Smid~\cite{NS-GSN-06} and the survey by Bose and Smid~\cite{BS11}.

One way to construct a geometric spanner is by using a Delaunay triangulation. The Delaunay triangulation is defined as follows: for any two vertices $u$ and $v$, if there exists a circle with $u$ and $v$ on its boundary and no other vertex in its interior, then the edge between $u$ and $v$ is part of the Delaunay triangulation. Equivalently, this can be defined using three vertices $u$, $v$, and $w$, where the triangle connecting these three vertices is part of the Delaunay triangulation if and only if the unique circle through these three points does not contain any other vertices in its interior. For simplicity, it is usually assumed that no three points are collinear and no four points lie on the boundary of the circle. 

The tight spanning ratio of the Delaunay triangulation is not known. Dobkin~\etal~\cite{DFS90} showed an upper bound of $\pi (1 + \sqrt{5}) / 2 \approx 5.09$ for the spanning ratio, which Keil and Gutwin~\cite{KG92} improved to $4 \pi / 3\sqrt{3} \approx 2.42$. Currently, the best upper bound is 1.998, proven by Xia~\cite{X13}. A lower bound on the spanning ratio was provided by Bose~\etal~\cite{BDLSV11}, who showed that this is strictly larger than $\tfrac{\pi}{2}$. This was later improved to 1.59 by Xia and Zhang~\cite{XZ11}.

Usually, the distance between two points $u$ and $v$ in the plane is defined as $((x_u - x_v)^2 + (y_u - y_v)^2)^{\tfrac{1}{2}}$. This distance can be generalized to a family of metrics $L_p$ where the distance between $u$ and $v$ is defined as $((x_u - x_v)^p + (y_u - y_v)^p)^{\tfrac{1}{p}}$. The shape of a ``circle'' varies in different metrics, leading to different Delaunay triangulations in different metrics. For example, the shape of the ``circle'' would be a diamond or square in the $L_1$ and $L_{\infty}$ metrics. 

In 1986, Lee and Lin~\cite{LL86} introduced the notion of generalized Delaunay triangulations. Instead of using a circle to construct the graph, generalized Delaunay triangulations can be constructed using arbitrary geometric shapes. It was proven that any generalized Delaunay triangulation constructed using a convex shape is a spanner~\cite{BCCS10}. 

Although generalized Delaunay triangulations using arbitrary convex shapes are known to be spanners, their spanning ratios are less well understood. Tight bounds on the spanning ratio are only known when an equilateral triangle, a square, or a regular hexagon is used in the construction. For equilateral triangles, Chew~\cite{C89} showed that the spanning ratio is~2. For squares, Chew~\cite{C86} showed an upper bound of~$\sqrt{10} \approx 3.16$, and Bonichon~\etal~\cite{BGHP15} showed a matching upper and lower bound of~$\sqrt{4 + 2 \sqrt{2}} \approx 2.61$. When using regular hexagons, Perkovi\'c~\etal~\cite{DBLP:journals/jocg/Perkovic0T21} showed a tight bound of~2. 

Bose~\etal~\cite{BCR18} studied generalized Delaunay triangulations using rectangles. For rectangles with aspect ratio~$A$, they showed upper and lower bounds of~$\sqrt{2}(2A+1)$ and $\sqrt{2} \sqrt{1+\A^2 + \A \sqrt{\A^2 + 1}}$, respectively. We obtain a tight upper bound of $\sqrt{2} \sqrt{1+\A^2 + \A \sqrt{\A^2 + 1}}$, matching the previous lower bound. To obtain this, we significantly extend and generalize the approach of Bonichon~\etal~\cite{BGHP15}. 
We extend the class of shapes for which a tight bound is known for the spanning ratio of generalized Delaunay triangulations.
We note that the proof of our result is not a straightforward extension of Bonichon~\etal~\cite{BGHP15}, as we cannot simply rotate our lemmas to get them to prove both the horizontal and vertical cases simultaneously.

\section{Preliminaries}
Let us first formally define the rectangle Delaunay triangulation of a set of points $P$. Given an arbitrary axis-aligned rectangle $R$, the rectangle Delaunay triangulation is constructed by considering scaled translates of $R$ (rotations are not allowed). Such scaled translates are also referred to as homothets. Given two vertices $u$ and $v$ in $P$, the rectangle Delaunay triangulation contains an edge between $u$ and $v$ if and only if there exists a scaled translate of $R$ with $u$ and $v$ on its boundary which contains no vertices of $P$ in the interior. Equivalently, the rectangle Delaunay triangulation contains a triangle $\triangle u v w$ if and only if there exists a scaled translate of $R$ with $u$, $v$, and $w$ on its boundary which contains no vertices of $P$ in the interior. We note that different rectangles can give different rectangle Delaunay triangulations. 

For our proofs, we assume that $P$ is in general position. Specifically, we assume that no four vertices lie on the boundary of any scaled translate of $R$ and that no two vertices lie on a line parallel to any of the sides of $R$ (i.e., no two vertices lie on a vertical or horizontal line). These assumptions are common for Delaunay graphs and are required to guarantee their planarity.

Throughout this paper, we use $\A$ to denote the aspect ratio of the rectangle $R$ used in the construction of the rectangle Delaunay triangulation, i.e., $\A = l/s$ where $l$ and $s$ are the length of the long and short side of $R$ respectively. We also use $d_{t}(u,v)$ to denote the length of the shortest path in the rectangle Delaunay triangulation between $u$ and $v$, $d_x(u,v)$ to denote the difference in $x$-coordinate between $u$ and $v$, $d_y(u,v)$ to denote the difference in $y$-coordinate between $u$ and $v$, and $d_2(u,v)$ to denote the Euclidean distance between $u$ and~$v$. 

\section{Bounding the Spanning Ratio} 
\label{sec:spanningRatio}
To show an upper bound on the spanning ratio between any two vertices $u$ and $v$, we consider the sequence of triangles $T_1, T_2, ..., T_k$ intersecting with line segment $uv$. The order of this sequence is determined by the order in which these triangles are encountered when following $uv$ from $u$ to $v$ (as shown in Figure~\ref{fig:triangles}). The boundary of each triangle except $T_1$ and $T_k$ intersects the interior of $uv$ twice. Hence, we can define the last line segment of $T_i$ ($1 \le i < k$) that intersects $uv$ as the line segment involved in the second intersection. We use $h_i$ and $l_i$ to denote the endpoints of the last line segment of $T_i$, where $h_i$ is the endpoint above $uv$ and $l_i$ is the endpoint below $uv$. Since all $T_i$ are triangles, we have that for every $T_i$ and $T_{i+1}$, either $l_i = l_{i+1}$ or $h_i = h_{i+1}$. We also define $h_0 = l_0 = u$, $l_k = h_k = v$. 

\begin{figure}[ht]
\centering
 \includegraphics[width=0.8\textwidth]{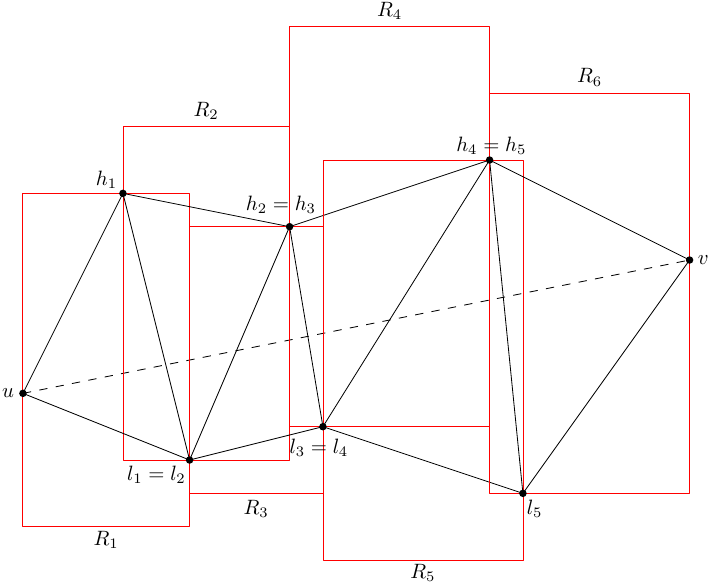}
\caption{The triangles intersecting $uv$ and their associated rectangles and $h_i$ and $l_i$.}
\label{fig:triangles}
\end{figure} 

Each triangle $T_i$ also has an associated rectangle $R_i$: the scaled translate of $R$ that has the three vertices of $T_i$ on its boundary. For ease of reference, we use W (west), N (north), E (east), and S (south) to refer to the four sides of a rectangle. We also use these sides to classify an edge, for example, if an endpoint of an edge lies on the W side of $R_i$ and the other endpoint lies on the N side of $R_i$, we call the edge a WN edge. We also define $u$ to be on the E side of $R_0$ (not associated with any triangle), as this will simplify some of the lemma statements. 

Define~$L$ to be the length of the vertical side of~$R$ divided by the length of the horizontal size of~$R$. Note that $L$ can be either $\A$ or $1/\A$. For our proofs, it is helpful to distinguish between edges of slope less than the slope of the diagonal of $R$ and those with larger slope. 

\begin{definition}
An edge is gentle if it has a slope within [-$L, L$]. Otherwise it is steep.
\end{definition}

We let $u,v$ be any two vertices in the rectangle Delaunay triangulation. Fix the $(x,y)$-coordinate system so that we have $L d_x(u,v) > d_y(u,v)$. Note that this is without loss of generality, since we can simply switch the $x$- and $y$-axes if needed. For instance, in Figure~\ref{fig:triangles}, let the $x$-axis be horizontal (and the $y$-axis be vertical) so that $Ld_x(u,v)>d_y(u,v)$ where $L=\A$. Therefore, if we consider a scaled translate of $R$ with $u$ lying in the lower left corner and passing through $v$, then $v$ lies on the E side. Without loss of generality, we assume $u$ to be at the origin $(0,0)$ and $v$ to be at $(x,y)$. We use $R(u,v)$ to denote the rectangle with $u$ and $v$ in opposite corners. 

In order to bound the spanning ratio of the rectangle Delaunay triangulation, we first define what it means for a rectangle to have potential. This later helps us to bound the total length of the shortest path between $u$ and $v$ in the rectangle Delaunay triangulation. 

\begin{definition}
The inductive point $c$ of a rectangle $R_i$ is the point with larger $x$-coordinate out of $h_i$ and $l_i$. Rectangle $R_i$ is inductive if edge ($l_i$, $h_i$) is gentle. 
\end{definition}

Let $d_{R_i}(h_{i},l_{i})$, $1\leq i\leq k$, denote the Euclidean distance when moving clockwise from $h_i$ to $l_i$ along the sides of $R_i$. 
\begin{definition}
A rectangle $R_i$ has potential if $d_{t} (u, h_{i}) + d_{t} (u, l_{i}) + d_{R_{i}}(h_{i},l_{i}) \le (2+2L) x_i$ where $x_i$ is the $x$-coordinate of the E side of $R_i$.
\end{definition}

We are now ready to prove that rectangles that are not inductive pass on their potential. 

\begin{lemma}
\label{lem:potential}
If $R(u,v)$ is empty and $(u,v)$ is not an edge in the rectangle Delaunay triangulation, then $R_1$ has potential. Furthermore, for any $1 \le i < k$, if $R_i$ has potential but is not inductive (i.e., $(l_{i}, h_{i})$ is steep), then $R_{i+1}$ has potential. 
\end{lemma}
\begin{proof}
Recall that $u$ is at the origin ($u=(0,0)$). By our general position assumption, $u$, $h_1$, and $l_1$ all lie on different sides of $R_1$. Since $R(u,v)$ is empty, $u$ cannot lie on the S side as this would mean that $l_i$ lies on the E side inside $R(u,v)$, similarly $u$ cannot lie on the N side or the E side. Hence, $u$ lies on the W side of $R_1$ and $x_1$ is the length of the horizontal side of $R_1$. Now that we know where $u$ lies, we can determine that $h_1$ lies on the N or E side of $R_1$ and $l_1$ lies on the E or S side of $R_1$. This implies that $d_{t}(u, h_{1}) + d_{t}(u,l_{1})+d_{R_1} (h_{1}, l_{1}$) is bounded by the perimeter of $R_1$ which is $(2+2L)x_1$. Thus, $R_1$ has potential, as claimed. 

Next, assume that $R_i$, for $1 \leq i < k$, has potential but is not inductive. Since $R_i$ is not inductive, we know that $(l_{i}, h_{i})$ is steep. In the remainder of this proof, we assume that $x_{l_i} < x_{h_i}$. The case where $x_{l_i} > x_{h_i}$ can be proven using analogous arguments.

Since  $x_{l_i} < x_{h_i}$, $l_i$ must be on the S side of $R_i$ and $h_i$ must be on the N or E side of $R_i$. If $h_i$ is on the N side of $R_i$, then because  $x_{l_i} < x_{h_i}$, $h_i$ must be on the N side of $R_{i+1}$ and $l_i$ must be on the S or W side of $R_{i+1}$. If $l_i$ is on the S side of $R_{i+1}$ (see Figure~\ref{fig:lemma4-case1}), we have  
\begin{equation}
d_{R_{i+1}}(h_i,l_i)-d_{R_i}(h_i,l_i) = 2(x_{i+1}-x_i). \label{eq:equation1}
\end{equation} 

\begin{figure}[ht]
  \begin{minipage}[b]{0.45\linewidth}
    \centering
    \includegraphics[width=0.7\textwidth]{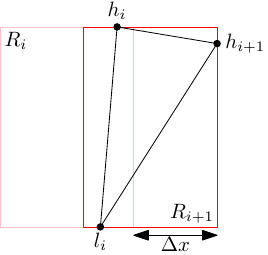}
    \caption{Vertex $h_i$ lies on the N side of $R_i$ and $l_i$ is on the S side of $R_{i+1}$.}
    \label{fig:lemma4-case1}
  \end{minipage}
  \hspace{0.5cm}
  \begin{minipage}[b]{0.45\linewidth}
    \centering
    \includegraphics[width=0.8\textwidth]{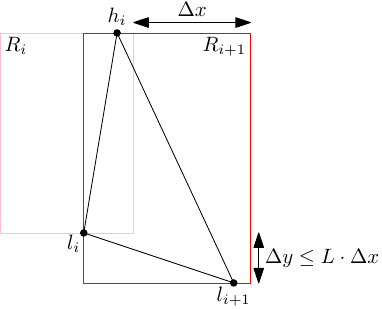}
    \caption{Vertex $h_i$ lies on the N side of $R_i$ and $l_i$ is on the W side of $R_{i+1}$.}
    \label{fig:lemma4-case2}
  \end{minipage}
\end{figure}

If $l_i$ is on the W side of $R_{i+1}$ (see Figure~\ref{fig:lemma4-case2}), let~$\Delta x = x_{i+1} - x_i$, and~$\Delta y = y_{l_{i}} - y_{l_{i+1}}$. We note that the intersection of $R_i$ and $R_{i+1}$ is a rectangle, whose long side has length more than $L$ times the short side's length, since $l_i$ is on the W side of $R_{i+1}$ and $h_i$ is on the N side of both $R_i$ and $R_{i+1}$. Thus, $\Delta y \leq L \cdot \Delta x$. Then we have that 
\begin{equation}
d_{R_{i+1}}(h_i,l_i)-d_{R_i}(h_i,l_i) 
\le
2 (\Delta x + \Delta y)
\le
(2+2L) \Delta x
=
(2+2L) (x_{i+1}-x_i). \label{eq:equation2}
\end{equation} 

If $h_i$ is on the E side of $R_i$, then because $x_{i+1} > x_i$, $h_i$ must be on the N side of $R_{i+1}$ and either $l_i$ is on the S side of $R_{i+1}$ and Equation~\ref{eq:equation1} holds or $l_i$ is on the W side of $R_{i+1}$ and Inequality~\ref{eq:equation2} holds.

It remains to show that the above inequalities imply that $R_{i+1}$ has potential. Since $R_i$ has potential, the above inequalities imply in all cases that: 
\begin{equation*} 
  d_{t}(u,h_{i}) + d_{t}(u,l_{i})+d_{R_{i+1}}(h_{i},l_{i}) \le (2+2L)x_{i+1}. 
\end{equation*}

\begin{figure}[ht]
\centering
\includegraphics{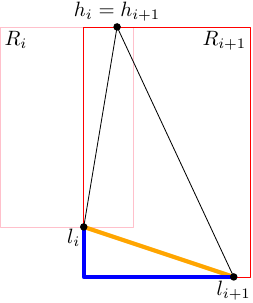}
\caption{Applying the triangle inequality when $h_i = h_{i+1}$.}
\label{fig:lemma4-3}
\end{figure} 

Since we have that either $l_i = l_{i+1}$ or $h_i = h_{i+1}$, the triangle inequality implies that $d_{t}(u,h_{i+1}) + d_{t}(u,l_{i+1})+d_{R_{i+1}}(h_{i+1},l_{i+1}) \le d_{t}(u,h_{i}) + d_{t}(u,l_{i})+d_{R_{i+1}}(h_{i},l_{i})$ (Figure~\ref{fig:lemma4-3} shows the case where $h_i = h_{i+1}$). This implies that 
\begin{equation*} 
d_{t}(u,h_{i+1}) + d_{t}(u,l_{i+1})+d_{R_{i+1}}(h_{i+1},l_{i+1}) \le (2+2L)x_{i+1}, 
\end{equation*}
completing the proof. 
\end{proof}

Next, we bound the distance from $u$ to the inductive point of a rectangle with potential when this inductive point lies on the E side of the rectangle. 

\begin{lemma}
\label{lem:inductiveEast}
If rectangle $R_i$ has potential and its inductive point $c$ ($c = h_i$ or $c = l_i$) lies on the E side of $R_i$, then $d_t (u,c) \le (1+L) x_c$. 
\end{lemma}
\begin{proof}
Assume  without loss of generality that $c = h_i$. Since $R_i$ has potential, $d_{t} (u, h_{i}) + d_{t} (u, l_{i})+d_{R_{i}}(h_{i},l_{i}) \le (2+2L) x_i$. This implies that $d_{t} (u, h_{i}) \le (1+L)x_{i}$ or $d_{t} (u, l_{i}) + d_{R_{i}}(h_{i},l_{i}) \le (1+L)x_{i}$. In the latter case, we use that $d_{t} (u, h_{i}) \le  d_{t} (u, l_{i}) + d_{2} (l_{i}, h_{i}) \le d_{t} (u, l_{i}) + d_{R_{i}}(h_{i},l_{i}) \le (1+L) x_i$ to get $d_{t} (u, h_{i}) \le (1+L)x_{i}$. Hence, in both cases we obtain that $d_{t} (u, h_{i}) \le (1+L)x_{i}$. Since $c$ is on the E side of $R_i$, $x_c = x_i$ and thus $d_t (u,c) \le (1+L) x_c$. 
\end{proof}

Now we shift our focus to paths consisting of gentle edges (see Figure~\ref{fig:maximalPath}). 

\begin{figure}[ht]
\centering
\includegraphics{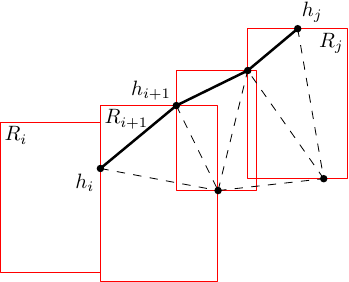}
\caption{An example of a maximal high path (thick edges). The other edges of the triangles are shown using dashed line segments.}
\label{fig:maximalPath}
\end{figure} 

\begin{definition}
\label{def:maximalPath}
If $h_j$ is on the E side of $R_j$, the maximal high path ending at $h_j$ is $h_j$ itself; otherwise, it is the path $h_i, h_{i+1}, ..., h_j$ such that $h_{m}$ is not on the E side of $R_{m}$ (for $i < m \leq j$) and either $i = 0$ or $h_i$ is on the E side of $R_i$. 

If $l_j$ is on the E side of $R_j$, the maximal low path ending at $l_j$ is $l_j$; otherwise, it is the path $l_i, l_{i+1}, ..., l_j$ such that $l_{m}$ is not on the E side of $R_{m}$ (for $i < m \leq j$) and either $i = 0$ or $l_i$ is on the E side of $R_i$. 
\end{definition}

Next, we bound the length of these maximal high and maximal low paths. 

\begin{lemma}
\label{lem:maximalPath}
If the path $h_i, h_{i+1}, ..., h_j$ is a maximal high path then $d_{t}(h_{i},h_{j}) \le (x_{h_j} - x_{h_i}) + (y_{h_j} - y_{h_i})$. Similarly, if the path $l_i, l_{i+1}, ..., l_j$ is a maximal low path then $d_{t}(l_{i},l_{j}) \le (x_{l_j} - x_{l_i}) + (y_{l_i} - y_{l_j})$. 
\end{lemma}
\begin{proof}
By Definition~\ref{def:maximalPath}, none of $h_{i+1}, ..., h_j$ are on the E sides of $R_{i+1}, ..., R_j$ respectively. This implies that all edges on a maximal high path are WN edges. By triangle inequality, $d_{t}(h_{k},h_{k+1}) \le (x_{h_{k+1}} - x_{h_k}) + (y_{h_{k+1}} - y_{h_{k}})$ for $i \leq k < j$. Summing up these terms, we obtain that $d_{t}(h_{i},h_{j}) \le (x_{h_j} - x_{h_i}) + (y_{h_j} - y_{h_i})$. An analogous argument proves that $d_{t}(l_{i},l_{j}) \le (x_{l_j} - x_{l_i}) + (y_{l_i} - y_{l_j})$ in a maximal low path. 
\end{proof}

We now use the above lemmas to prove bounds on the path length from $u$ to the inductive point on the first inductive rectangle (if one exists) when $R(u,v)$ does not contain any vertices. Note that in Property~2 of Lemma~\ref{lem:firstInductive}, we differentiate between~$L=A$ and~$L=1/A$, which is crucial in proving Theorem~\ref{thm:upperBound1}.

\begin{lemma}
\label{lem:firstInductive}
Let $R(u,v)$ not contain any vertices of $P$ and let $(u,v)$ not be an edge of the rectangle Delaunay triangulation. The following properties hold: 
\begin{enumerate}
\item If no rectangle in $R_1, ...,R_k$ is inductive then 
\[d_{t}(u, v) \le (L + \sqrt{L^{2}+1})x+y.\] 
\item Otherwise, let $R_j$ be the first inductive rectangle in the sequence $R_1, ..., R_k$. 
\begin{enumerate}
\item If $h_j$ is the inductive point of $R_j$ and $L=\A$, then 
\[ d_{t}(u,h_{j})+ (y_{h_{j}}-y) \le (\A+\sqrt{\A^2+1})x_{h_j}.\]
\item If $h_j$ is the inductive point of $R_j$ and $L=\tfrac{1}{\A}$, then
\[ d_{t}(u,h_{j})+ \A(y_{h_{j}}-y) \le \left(1+\sqrt{\tfrac{1}{\A^2}+1} \right)x_{h_j}.\]
\item If $l_j$ is the inductive point of $R_j$ and $L=\A$, then
\[d_{t}(u,l_{j}) - y_{l_j} \le (\A + \sqrt{\A^2+1})x_{l_j}.\]
\item If $l_j$ is the inductive point of $R_j$ and $L=\tfrac{1}{\A}$, then
\[d_{t}(u,l_{j}) - \A y_{l_j} \le \left(1 + \sqrt{\tfrac{1}{\A^2}+1} \right)x_{l_j}.\]
\end{enumerate}
\end{enumerate}
\end{lemma}
\begin{proof}
\emph{Property 1:} By Lemma~\ref{lem:potential}, if no rectangle in $R_1, ..., R_k$ is inductive then the last rectangle must have potential since $R_1$ has potential. Since no two vertices have the same $y$-coordinate, $v$ must lie on the E side of the last rectangle. Thus, we can use Lemma~\ref{lem:inductiveEast} to conclude that $d_{t}(u, v) \le (1+L)x \le (L + \sqrt{L^{2}+1})x+y.$ 

\emph{Property 2a:} We consider the situation where $R_j$ is the first inductive rectangle in the sequence $R_1, ..., R_k$. Let $l_i, ..., l_{j-1} = l_j$ be the maximal low path ending at $l_j$, and recall that $h_j$ is the inductive point of $R_j$. By Lemma~\ref{lem:potential} we know that $R_i$ has potential, since $R_1$ has potential and no rectangle before $R_i$ is inductive. Since $R_i$ has potential and $l_i$ is on the E side of $R_i$, by Lemma~\ref{lem:inductiveEast} we know $d_{t}(u,l_{i}) \le (1+L)x_{l_i}$. See Figure~\ref{fig:lemma8}. Since $L=\A$, we have
\begin{align*}
d_{t}(u,h_{j}) + (y_{h_{j}}-y) &\le d_{t}(u,l_{i}) + d_{t}(l_{i},l_{j}) + d_{2}(l_{j},h_{j}) + (y_{h_{j}}-y)\\
&\le (1+\A)x_{l_i} + d_{t}(l_{i},l_{j}) + d_{2}(l_{j},h_{j}) + y_{h_{j}}.
\end{align*}

\begin{figure}[htb]
    \centering
    \includegraphics[width=.8\textwidth]{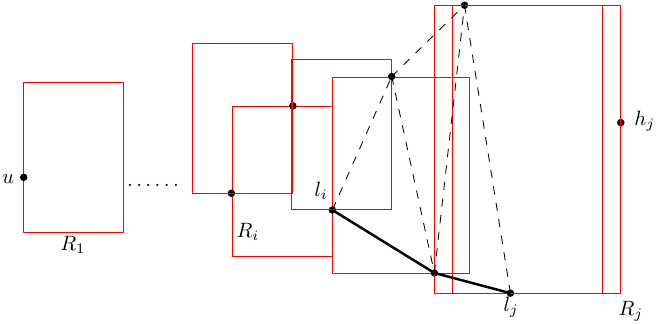}
    \caption{In Property $2a$, $R_j$ is the first inductive rectangle, $h_j$ is the inductive point of $R_j$, and $l_i, ..., l_{j-1} = l_j$ is the maximal low path ending at $l_j$. $R_i$ has potential and $l_i$ is on the E side of $R_i$.}
    \label{fig:lemma8}
\end{figure}

Since $l_i, ..., l_{j-1} = l_j$ is a maximal low path, by Lemma~\ref{lem:maximalPath} we know $d_{t}(l_{i},l_{j}) \le (x_{l_j} - x_{l_i}) + (y_{l_i} - y_{l_j})$. Hence, we obtain that:
\begin{align*}
d_{t}(u,h_{j}) + (y_{h_{j}}-y) &\le (1+\A)x_{l_i} + (x_{l_j} - x_{l_i}) + (y_{l_i} - y_{l_j}) + d_{2}(l_{j},h_{j}) + y_{h_{j}} \\
&= \A x_{l_i} + x_{l_j} + (y_{l_i} - y_{l_j}) + d_{2}(l_{j},h_{j}) + y_{h_{j}}.
\end{align*}

Because $R_j$ is inductive, we know that edge $(l_j, h_j)$ is gentle. Therefore, $y_{h_j}-y_{l_j}<\A (x_{h_j}-x_{l_j})$ and $d_{2}(l_{j},h_{j}) \le \sqrt{1+\A^2}(x_{h_j}-x_{l_j})$. Thus:
\begin{align*}
d_{t}(u,h_{j}) + (y_{h_{j}}-y) &\le \A x_{l_i} + x_{l_j} + (y_{l_i} - y_{l_j}) + \sqrt{1+\A^2}(x_{h_j}-x_{l_j}) + y_{h_{j}} \\
&\le \A x_{l_i} + (y_{l_i} - y_{l_j}) + \sqrt{1+\A^2} x_{h_j} + y_{h_{j}}.
\end{align*}

Furthermore, again because edge $(l_j, h_j)$ is gentle, we have that $y_{h_j} - y_{l_j} \le \A(x_{h_j}-x_{l_j})$ and therefore: 
\begin{align*}
d_{t}(u,h_{j}) + (y_{h_{j}}-y) &\le \A x_{l_i} + y_{l_i} + \sqrt{1+\A^2} x_{h_j} + \A(x_{h_j}-x_{l_j}) \\
&\le y_{l_i} + \sqrt{1+\A^2} x_{h_j} + A x_{h_j}.
\end{align*}

Note that when $i\leq j$, then~$x_{l_i}\leq x_{l_j}$. Since $R(u,v)$ is empty, $l_i$ must lie below it. Thus $y_{l_i} < 0$, which leads to $d_{t}(u,h_{j}) + (y_{h_{j}}-y) \le (\A + \sqrt{\A^2 + 1}) x_{h_j}$, as required.

\emph{Property 2b:} Let $R_j$ be the first inductive rectangle in the sequence $R_1, ..., R_k$. Let $l_i, ..., l_{j-1} = l_j$ be the maximal low path ending at $l_j$, and recall that $h_j$ is the inductive point of $R_j$. By Lemma~\ref{lem:potential}, $R_i$ has potential, and by Lemma~\ref{lem:inductiveEast}, we have $d_{t}(u,l_{i}) \le (1+L)x_{l_i}$. Since $L=\tfrac{1}{\A}$, we have
\begin{align*}
d_{t}(u,h_{j}) + \A(y_{h_{j}}-y) &\le d_{t}(u,l_{i}) + d_{t}(l_{i},l_{j}) + d_{2}(l_{j},h_{j}) + \A(y_{h_{j}}-y)\\
&\le (1+\tfrac{1}{\A})x_{l_i} + d_{t}(l_{i},l_{j}) + d_{2}(l_{j},h_{j}) + \A y_{h_{j}}.
\end{align*}

Since $l_i, ..., l_{j-1} = l_j$ is a maximal low path, by Lemma~\ref{lem:maximalPath} we know $d_{t}(l_{i},l_{j}) \le (x_{l_j} - x_{l_i}) + (y_{l_i} - y_{l_j})$. Because $R_j$ is inductive, we know that edge $(l_j, h_j)$ is gentle. Therefore, $d_{2}(l_{j},h_{j}) \le \sqrt{1+\tfrac{1}{\A^2}}(x_{h_j}-x_{l_j})$ and thus:
\begin{align*}
d_{t}(u,h_{j}) + \A(y_{h_{j}}-y) 
&\le (1+\tfrac{1}{\A})x_{l_i} + (x_{l_j}-x_{l_i}) + (y_{l_i}-y_{l_j}) + d_{2}(l_{j},h_{j}) + \A y_{h_{j}} \\
&\le (1+\tfrac{1}{\A})x_{l_i} + (x_{l_j}-x_{l_i}) + (y_{l_i}-y_{l_j}) + \sqrt{1+\tfrac{1}{\A^2}}(x_{h_j}-x_{l_j}) + \A y_{h_{j}} \\
&\le \tfrac{1}{\A}x_{l_i} + (y_{l_i}-y_{l_j}) + \sqrt{1+\tfrac{1}{\A^2}}x_{h_j} + \A y_{h_{j}}.
\end{align*}

Again because edge $(l_j, h_j)$ is gentle, we have that $y_{h_j} - y_{l_j} \le \tfrac{1}{\A}(x_{h_j}-x_{l_j})$. Therefore $\A(y_{h_j} - y_{l_j}) \le (x_{h_j}-x_{l_j})$. We have $\A \geq 1$ and therefore:
\begin{align*}
d_{t}(u,h_{j}) + \A(y_{h_{j}}-y) &\le \tfrac{1}{\A}x_{l_i} + \A (y_{l_i}-y_{l_j}) + \sqrt{1+\tfrac{1}{\A^2}}x_{h_j} + \A y_{h_{j}} \\
&\le \tfrac{1}{\A} x_{l_i} + \A y_{l_i} + (x_{h_j}-x_{l_j}) + \sqrt{1+\tfrac{1}{\A^2}}x_{h_j}.
\end{align*}

Since $\tfrac{1}{\A}\le 1$ and~$i \leq j$, we have $\tfrac{1}{\A} x_{l_i}\le x_{l_i}\le x_{l_j}$. Therefore
\begin{align*}
d_{t}(u,h_{j}) + \A(y_{h_{j}}-y) \le \A y_{l_i} + x_{h_j} + \sqrt{1+\tfrac{1}{\A^2}}x_{h_j}
\le (1+\sqrt{\tfrac{1}{\A^2}+1})x_{h_j}.
\end{align*}

\emph{Property 2c:} Let $R_j$ be the first inductive rectangle in the sequence $R_1, ..., R_k$. Now, let $h_i, ..., h_{j-1} = h_j$ be the maximal high path ending at $h_j$, and recall that $l_j$ is the inductive point of $R_j$. By Lemma~\ref{lem:potential}, $R_i$ has potential, and by Lemma~\ref{lem:inductiveEast}, we have $d_{t}(u,h_{i}) \le (1+L)x_{h_i}$. Since $L=\A$, 
\begin{align*}
d_{t}(u,l_j) - y_{l_j} &\le d_{t}(u,h_i) + d_{t}(h_i,h_j) + d_2(h_j,l_j) - y_{l_j} \\
&\le (1+\A)x_{h_i} + d_{t}(h_i,h_j) + d_2(h_j,l_j) - y_{l_j}. 
\end{align*}

Since $h_i, ..., h_{j-1} = h_j$ is a maximal high path, by Lemma~\ref{lem:maximalPath} we know $d_{t}(h_{i},h_{j}) \le (x_{h_j} - x_{h_i}) + (y_{h_j} - y_{h_i})$. It follows that:
\begin{align*}
d_{t}(u,l_j) - y_{l_j} &\le (1+\A)x_{h_i} + (x_{h_j} - x_{h_i}) + (y_{h_j} - y_{h_i}) + d_2(h_j,l_j) - y_{l_j} \\
&= \A x_{h_i} + x_{h_j} + (y_{h_j} - y_{h_i}) + d_2(h_j,l_j) - y_{l_j}.
\end{align*}

Because $R_j$ is inductive, we know that edge $(l_j, h_j)$ is gentle. Therefore, $d_{2}(h_{j},l_{j}) \le \sqrt{1+\A^2}(x_{l_j}-x_{h_j})$ and thus:
\begin{align*}
d_{t}(u,l_j) - y_{l_j} &\le \A x_{h_i} + x_{h_j} + (y_{h_j} - y_{h_i}) + \sqrt{1+\A^2}(x_{l_j}-x_{h_j}) - y_{l_j} \\
&\le \A x_{h_i} + (y_{h_j} - y_{h_i}) + \sqrt{1+\A^2}x_{l_j} - y_{l_j}.
\end{align*}

Furthermore, again because edge $(l_j, h_j)$ is gentle, we have that $y_{h_j} - y_{l_j} \le \A(x_{l_j}-x_{h_j})$ and therefore: 
\begin{align*}
d_{t}(u,l_j) - y_{l_j} &\le \A x_{h_i} - y_{h_i} + \sqrt{1+\A^2}x_{l_j} + \A(x_{l_j}-x_{h_j}) \\
&\le - y_{h_i} + \sqrt{1+\A^2}x_{l_j} + \A x_{l_j}.
\end{align*}

Since $R(u,v)$ is empty, $h_i$ must lie above it. Thus $y_{h_i} > 0$, which leads to $d_{t}(u,l_j) - y_{l_j} \le (\A + \sqrt{\A^2+1})x_{l_j}$, as required.

\emph{Property 2d:} Let $R_j$ be the first inductive rectangle in the sequence $R_1, ..., R_k$. Now, let $h_i, ..., h_{j-1} = h_j$ be the maximal high path ending at $h_j$, and recall that $l_j$ is the inductive point of $R_j$. By Lemma~\ref{lem:potential}, $R_i$ has potential, and by Lemma~\ref{lem:inductiveEast}, we have $d_{t}(u,h_{i}) \le (1+L)x_{h_i}$.  Since $L=\tfrac{1}{\A}$, 
\begin{align*}
d_{t}(u,l_j) - \A y_{l_j} &\le d_{t}(u,h_i) + d_{t}(h_i,h_j) + d_2(h_j,l_j) - \A y_{l_j} \\
&\le (1+\tfrac{1}{\A})x_{h_i} + d_{t}(h_i,h_j) + d_2(h_j,l_j) - \A y_{l_j}.
\end{align*}

Since $h_i, ..., h_{j-1} = h_j$ is a maximal high path, by Lemma~\ref{lem:maximalPath} we know $d_{t}(h_{i},h_{j}) \le (x_{h_j} - x_{h_i}) + (y_{h_j} - y_{h_i})$. Because edge $(l_j,h_j)$ is gentle,  we have that $d_2(h_j,l_j) \le \sqrt{1+\tfrac{1}{\A^2}}(x_{l_j}-x_{h_j})$. It follows that:
\begin{align*}
d_{t}(u,l_j) - \A y_{l_j} &\le (1+\tfrac{1}{\A})x_{h_i} + (x_{h_j}-x_{h_i}) + (y_{h_j}-y_{h_i}) + d_2(h_j,l_j) - \A y_{l_j}\\
&\le (1+\tfrac{1}{\A})x_{h_i} + (x_{h_j}-x_{h_i}) + (y_{h_j}-y_{h_i}) + \sqrt{1+\tfrac{1}{\A^2}}(x_{l_j}-x_{h_j}) - \A y_{l_j} \\
&\le \tfrac{1}{\A}x_{h_i} + (y_{h_j}-y_{h_i}) + \sqrt{1+\tfrac{1}{\A^2}}x_{l_j} - \A y_{l_j}.
\end{align*}

Again because edge $(l_j, h_j)$ is gentle, we have that $y_{h_j} - y_{l_j} \le \tfrac{1}{\A}(x_{l_j}-x_{h_j})$. Therefore $\A(y_{h_j} - y_{l_j}) \le (x_{l_j}-x_{h_j})$ and \begin{align*}
d_{t}(u,l_j) - \A y_{l_j} &\le \tfrac{1}{\A}x_{h_i} + \A(y_{h_j}-y_{h_i}) + \sqrt{1+\tfrac{1}{\A^2}}x_{l_j} - \A y_{l_j} \\
&\le \tfrac{1}{\A}x_{h_i} + (x_{l_j}-x_{h_j}) - \A y_{h_i} + \sqrt{1+\tfrac{1}{\A^2}}x_{l_j}.
\end{align*}

Since $\tfrac{1}{\A}\le 1$, we have $\tfrac{1}{\A}x_{h_i}\le x_{h_i}\le  x_{h_j}$. Thus
\begin{align*}
d_{t}(u,l_j) - \A y_{l_j} &\le x_{l_j} - \A y_{h_i} + \sqrt{1+\tfrac{1}{\A^2}}x_{l_j} \\
&\le (1+\sqrt{\tfrac{1}{\A^2}+1})x_{l_j}.
\end{align*}
as required, completing our proof of Property 1, $2a$, $2b$, $2c$ and $2d$.
\end{proof}

Our final ingredient determines the types of edges we can encounter when the $y$-coordinate of a vertex differs significantly from that of $v$. Recall that $v=(x,y)$. 

\begin{lemma}
\label{lem:NEDown}
Let $R(u,v)$ not contain any vertices of $P$ and let the coordinates of the inductive point $c$ of $R_i$ be such that it satisfies $0 < L (x - x_c) < |y - y_{c}|$. 
\begin{itemize}
\item If  $c = h_i$ and thus $0 < L (x - x_c) < y_{c} - y$, then let $j$ be the smallest index larger than $i$ such that $L (x - x_{h_j}) \ge y_c - y \ge 0$. All edges on the path $h_i, ..., h_{j}$ are NE edges. 
\item If $c = l_i$ and thus $0 < L (x - x_c) < y - y_{c}$, then let $j$ be the smallest index larger than $i$ such that $L (x - x_{l_j}) \ge y - y_c \ge 0$. All edges on the path $l_i, ..., l_{j}$ are SE edges. 
\end{itemize}

\end{lemma}

\begin{proof}
We consider the case where $c = h_i$. The case where $c = l_i$ can be proven using an analogous argument. We first observe that there exists a $j$ satisfying $L (x - x_{h_j}) \ge y_c - y \ge 0$, since if we pick $v = h_j$, then $L (x - x_{h_j}) = y - y_c = 0$. 

In the remainder, let $j$ be the smallest index larger than $i$ such that $L (x - x_{h_j}) \ge y_c - y \ge 0$. Let $h_{m}$ be any point on the path $ h_i, ..., h_{j-1}$ and thus we have that $0 < L (x - x_{h_{m}}) < y_{h_{m}} - y$. We observe that edge $(h_m, h_{m+1})$ must be a NE, WN, or WE edge in $R_{m+1}$, since $x_{h_{m}} < x_{h_{m+1}}$ and neither vertex is below $uv$. However, if $h_{m}$ lies on the W side of $R_{m+1}$, then $v$ would be inside $R_{m+1}$ because $y_{l_m} < y$ and $0 < L (x - x_{h_{m}}) < y_{h_{m}} - y < y_{h_{m}} - y_{l_{m}}$. Thus, $(h_{m},h_{m+1})$ must be a NE edge. 
\end{proof}

We now have all the ingredients needed to prove our main result. 
Recall that, up to Lemma~\ref{lem:NEDown}, the $(x,y)$-coordinate system is fixed so that $L d_x(u,v) > d_y(u,v)$, i.e. $Lx \geq y$. However, for ease of exposition, in Theorem~\ref{thm:upperBound1} we instead fix the $(x,y)$-coordinate system so that all the homothet rectangles have their vertical sides being the long sides. 

Note that in Lemma~\ref{lem:firstInductive}, we obtain different upper bounds depending on whether~$L=A$ or~$L=1/A$. These two cases must be treated differently for the inductive proof of Theorem~\ref{thm:upperBound1} to hold. In particular, in Theorem~\ref{thm:upperBound1} the bound for $\A d_x(u,v) \ge d_y(u,v)$ does not coincide with the rotated version of the bound for $\A d_x(u,v) < d_y(u,v)$.

\begin{theorem}
\label{thm:upperBound1}
Let $u, v$ be any two vertices in the rectangle Delaunay triangulation. If $\A d_x(u,v) \ge d_y(u,v)$, then \[ d_t(u,v) \le (\A +\sqrt{\A^2+1})x + y. \]  
Otherwise, \[ d_t(u,v) \le \A x + \left( 1 +\sqrt{\tfrac{1}{\A ^2}+1} \right)y. \] 
\end{theorem}
\begin{proof}
We consider all pairs of vertices $(u,v)$ and order them by the size of the smallest scaled translate of $R$ that has both $u$ and $v$ on its boundary. We perform induction based on the rank in this ordering. 

We begin with the base case. Consider the first pair $(u,v)$ in this ordering, which has the smallest overall scaled translate of $R$. This smallest rectangle contains no vertices of $P$, as any such vertex would imply the existence of a smaller rectangle with two vertices on its boundary, contradicting that we are considering the smallest one. Since the rectangle is empty, by construction there is an edge between $u$ and $v$. Thus $d_t(u,v) = d_2(u,v) \leq x + y$, so the base case of the inductive hypothesis is satisfied. Note that the base case holds regardless of whether or not $\A d_x(u,v) \ge d_y(u,v)$. 

Next, consider an arbitrary pair $(u,v)$ and assume the theorem holds for all pairs $(u,v)$ defining a smaller rectangle. We consider two cases: $R(u,v)$ does not contain any vertex of $P$, and $R(u,v)$ contains some vertices of $P$.

\emph{Case 1}: There are no vertices inside $R(u,v)$. We distinguish two subcases, either $\A x\ge y$ or $\A x < y$.

\emph{Subcase $\A x\ge y$}: 
Note that since the vertical side of the homothets is the longer side, for the $(x,y)$-coordinate system we have $L=A$, and $Lx \geq y$. 

If $(u,v)$ is an edge in the rectangle Delaunay triangulation, then $d_t(u,v) \le x + y \le (\A +\sqrt{\A^2+1})x + y$. Otherwise, if no rectangle in $R_1, ..., R_k$ is inductive then by Property~1 of Lemma~\ref{lem:firstInductive} we know $d_t(u,v) \le (\A +\sqrt{\A^2+1})x + y$.

Hence, we focus on the case where there is an inductive rectangle. Let $R_i$ be the first inductive rectangle in the sequence $R_1, ..., R_k$. We distinguish the case where the inductive point is $h_i$ and where it is $l_i$. If $h_i$ is the inductive point of $R_i$ then by Property $2a$ of Lemma~\ref{lem:firstInductive} we know $d_{t}(u,h_{i})+ (y_{h_{i}}-y) \le (\A+\sqrt{\A^2+1})x_{h_i}$ and thus $d_{t}(u,h_{i}) \le (\A+\sqrt{\A^2+1})x_{h_i} -  (y_{h_{i}}-y)$.

If $\A (x - x_{h_i}) \ge y_{h_i} - y \ge 0$, we let $h_j = h_i$ in the remainder. Otherwise, we let $j$ be the smallest index larger than $i$ such that $\A (x - x_{h_j}) \ge y_{h_j} - y \ge 0$. By Lemma~\ref{lem:NEDown}, $h_j$ exists and all edges on the path $h_i, ..., h_{j}$ are NE edges. By triangle inequality, $d_{t}(h_{m},h_{m+1}) \le (x_{h_{m+1}} - x_{h_m}) + (y_{h_m} - y_{h_{m+1}})$ for any $h_m$ and $h_{m+1}$ on this path. This implies that $d_t(h_i,h_j) \leq (x_{h_j}-x_{h_i}) + (y_{h_i}-y_{h_j})$. Since $\A (x - x_{h_j}) \ge y_{h_j} - y \ge 0$ and the smallest scaled translate of $R$ with $h_j$ and $v$ on its boundary is smaller than that of $u$ and $v$, we can use induction to get $d_t(h_j,v) \le (\A +\sqrt{\A^2+1})d_x(h_j,v) + d_y(h_j,v)$. Putting everything together, we obtain that
\begin{align*}
d_t(u,v) &\le d_t(u,h_i) + d_t(h_i, h_j) + d_t(h_j,v) \\
&\le  (\A+\sqrt{\A^2+1})x_{h_i} -  (y_{h_{i}}-y) + (x_{h_j}-x_{h_i}) + (y_{h_i}-y_{h_j}) \\
&\quad +(\A +\sqrt{\A^2+1})d_x(h_j,v) + d_y(h_j,v) \\
&= (\A+\sqrt{\A^2+1})x_{h_i}  + (x_{h_j}-x_{h_i}) + (y-y_{h_j}) \\
&\quad +(\A +\sqrt{\A^2+1})d_x(h_j,v) + d_y(h_j,v) \\
&\le (\A +\sqrt{\A^2+1})d_x(u,h_j) - d_y(h_j,v) + (\A +\sqrt{\A^2+1})d_x(h_j,v) + d_y(h_j,v) \\
&= (\A +\sqrt{\A^2+1})x.
\end{align*}
proving the theorem when $h_i$ is the inductive point of $R_i$. 

If $l_i$ is the inductive point of $R_i$ then by Property $2c$ of Lemma~\ref{lem:firstInductive} we know $d_{t}(u,l_{i}) - y_{l_i} \le (\A + \sqrt{\A^2+1})x_{l_i}$ and thus $d_{t}(u,l_{i}) \le (\A + \sqrt{\A^2+1})x_{l_i} + y_{l_i}$.

If $\A (x - x_{l_i}) \ge y - y_{l_i}$, we let $l_j = l_i$ in the remainder. Otherwise, we let $j$ be the smallest index larger than $i$ such that $\A (x - x_{l_j}) \ge y - y_c \ge 0$. By Lemma~\ref{lem:NEDown}, $l_j$ exists and all edges on the path $l_i, ..., l_{j}$ are SE edges. By triangle inequality, $d_{t}(l_{m},l_{m+1}) \le (x_{l_{m+1}} - x_{l_m}) + (y_{l_{m+1}} - y_{h_{m}})$ for any $l_m$ and $l_{m+1}$ on this path. This implies that $d_t(l_i,l_j) \leq (x_{l_j}-x_{l_i}) + (y_{l_j}-y_{l_i})$. Since $\A (x - x_{l_j}) \ge y - y_{l_j} \ge 0$ and the smallest scaled translate of $R$ with $l_j$ and $v$ on its boundary is smaller than that of $u$ and $v$, we can use induction to get $d_t(l_j,v) \le (\A +\sqrt{\A^2+1})d_x(l_j,v) + d_y(l_j,v)$. Putting everything together, this implies that
\begin{align*}
d_t(u,v) &\le d_t(u,l_i) + d_t(l_i, l_j) + d_t(l_j,v) \\
&\le (\A + \sqrt{\A^2+1})x_{l_i} + y_{l_i} + (x_{l_j}-x_{l_i}) + (y_{l_j}-y_{l_i}) \\
&\quad + (\A +\sqrt{\A^2+1})d_x(l_j,v) + d_y(l_j,v) \\
&\le (\A + \sqrt{\A^2+1})x_{l_i} + (x_{l_j}-x_{l_i}) + y_{l_j} + (\A +\sqrt{\A^2+1})d_x(l_j,v) + d_y(l_j,v) \\
&\le (\A +\sqrt{\A^2+1})x + y.
\end{align*}
completing the proof of Case 1 when $\A x \geq y$.

\emph{Subcase $\A x < y$}:
Consider the $(x',y')$-coordinate system where the $x'$-axis equals the $y$-axis and the $y'$-axis equals the $x$-axis. See Figure~\ref{fig:swap_axes}. When we look at the homothet rectangles $R_1,\ldots,R_k$ intersecting the segment $uv$ in the $(x',y')$-coordinate system, the horizontal side of the homothets is the longer side and we have $L = 1/A$. Therefore, $Ax < y$ implies $Lx' > y'$.

\begin{figure}[ht]
    \centering
    \includegraphics{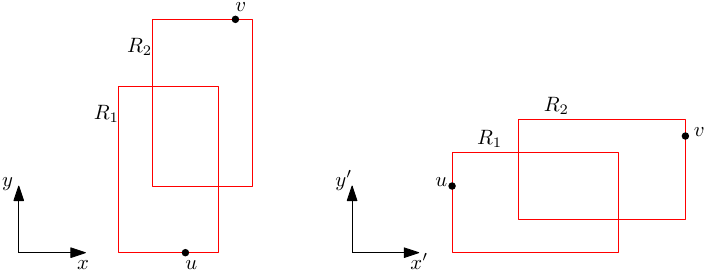}
    \caption{Homothet rectangles $R_1,\ldots,R_k$ in the $(x',y')$-coordinate system, for $k=2$.}
    \label{fig:swap_axes}
\end{figure}

If $(u,v)$ is an edge in the rectangle Delaunay triangulation, then $d_t(u,v) \le x' + y' \le (1 +\sqrt{\tfrac{1}{\A^2}+1})y + \A x$. If no rectangle in $R_1, ..., R_k$ is inductive then by Property~1 of Lemma~\ref{lem:firstInductive} we know $d_t(u,v) \le (\tfrac{1}{\A} +\sqrt{\tfrac{1}{\A^2}+1})x' + y'\le (1 +\sqrt{\tfrac{1}{\A^2}+1})y + \A x$.

When there is an inductive rectangle, define $R_i$, $h_i$ and $l_i$ as above. If $h_i$ is the inductive point of $R_i$ then by Property $2b$ of Lemma~\ref{lem:firstInductive} we know $d_{t}(u,h_{i})+ \A(y'_{h_{i}}-y') \le (1+\sqrt{\tfrac{1}{\A^2}+1})x'_{h_i}$.

If $\tfrac{1}{\A} (x' - x'_{h_i}) \ge y'_{h_i} - y' \ge 0$, we let $h_j = h_i$ in the remainder. Otherwise, we let $j$ be the smallest index larger than $i$ such that $\tfrac{1}{\A} (x' - x'_{h_j}) \ge y'_{h_j} - y' \ge 0$. By Lemma~\ref{lem:NEDown}, $h_j$ exists and all edges on the path $h_i, ..., h_{j}$ are NE edges. By triangle inequality, $d_t(h_i,h_j) \leq (x'_{h_j}-x'_{h_i}) + (y'_{h_i}-y'_{h_j})$. Since $\tfrac{1}{\A} (x' - x'_{h_j}) \ge y'_{h_j} - y' \ge 0$ and the smallest scaled translate of $R$ with $h_j$ and $v$ on its boundary is smaller than that of $u$ and $v$, we can use induction to get $d_t(h_j,v) \le (1 +\sqrt{\tfrac{1}{\A^2}+1})d_y(h_j,v) + \A d_x(h_j,v)$. Putting everything together, we obtain

\begin{align*}
d_t(u,v) &\le d_t(u,h_i) + d_t(h_i, h_j) + d_t(h_j,v) \\
&\le  (1+\sqrt{\tfrac{1}{\A^2}+1})x'_{h_i} -  \A(y'_{h_{i}}-y') + (x'_{h_j}-x'_{h_i}) + (y'_{h_i}-y'_{h_j}) \\
&\quad +(1 +\sqrt{\tfrac{1}{\A^2}+1})d_y(h_j,v) + \A d_x(h_j,v) \\
&\le  (1+\sqrt{\tfrac{1}{\A^2}+1})x'_{h_i} -  \A(y'_{h_{i}}-y') + (x'_{h_j}-x'_{h_i}) + \A (y'_{h_i}-y'_{h_j})  \\
&\quad +(1 +\sqrt{\tfrac{1}{\A^2}+1})d_y(h_j,v) + \A d_x(h_j,v) \\
&\le (1+\sqrt{\tfrac{1}{\A^2}+1})x'_{h_i}  + (x'_{h_j}-x'_{h_i}) - \A d_{y'}(h_j,v) \\
&\quad +(1 +\sqrt{\tfrac{1}{\A^2}+1})d_y(h_j,v) + \A d_x(h_j,v).
\end{align*}

Recall that the $y'$-axis in the $(x',y')$-coordinate system equals the $x$-axis in the $(x,y)$-coordinate system, so $\A d_{y'}(h_j,v)=\A d_{x}(h_j,v)$. Thus
\begin{align*}
d_t(u,v) &\le (1+\sqrt{\tfrac{1}{\A^2}+1})x'_{h_j} 
+(1 +\sqrt{\tfrac{1}{\A^2}+1})d_y(h_j,v) \\
&= (1+\sqrt{\tfrac{1}{\A^2}+1})y.
\end{align*}

If $l_i$ is the inductive point of $R_i$ then by Property $2d$ of Lemma~\ref{lem:firstInductive} we know $d_{t}(u,l_{i}) -  \A y'_{l_i} \le (1+\sqrt{\tfrac{1}{\A^2}+1})x'_{l_i}$. Thus $d_{t}(u,l_{i}) \le (1+\sqrt{\tfrac{1}{\A^2}+1})x'_{l_i} + \A y'_{l_i}$.

If $\tfrac{1}{\A} (x' - x'_{l_i}) \ge y' - y'_{l_i} \ge 0$, we let $l_j = l_i$ in the remainder. Otherwise, we let $j$ be the smallest index larger than $i$ such that $\tfrac{1}{\A} (x' - x'_{l_j}) \ge y' - y'_{l_j} \ge 0$. By Lemma~\ref{lem:NEDown}, $l_j$ exists and all edges on the path $l_i, ..., l_{j}$ are SE edges. By triangle inequality, $d_t(l_i,l_j) \leq (x'_{l_j}-x'_{l_i}) + (y'_{l_j}-y'_{l_i})$. Since $\tfrac{1}{\A} (x' - x'_{l_j}) \ge y' - y'_{l_j} \ge 0$ and the smallest scaled translate of $R$ with $l_j$ and $v$ on its boundary is smaller than that of $u$ and $v$, we can use induction to get $d_t(l_j,v) \le (1 +\sqrt{\tfrac{1}{\A^2}+1})d_y(l_j,v) + \A d_x(l_j,v)$. Thus we obtain that 

\begin{align*}
d_t(u,v) &\le d_t(u,l_i) + d_t(l_i, l_j) + d_t(l_j,v) \\
&\le  (1+\sqrt{\tfrac{1}{\A^2}+1})x'_{l_i} + \A y'_{l_i} + (x'_{l_j}-x'_{l_i}) + (y'_{l_j}-y'_{l_i}) \\
&\quad +(1 +\sqrt{\tfrac{1}{\A^2}+1})d_y(l_j,v) + \A d_x(l_j,v) \\
&\le  (1+\sqrt{\tfrac{1}{\A^2}+1})x'_{l_i} + \A y'_{l_i} + (x'_{l_j}-x'_{l_i}) +\A (y'_{l_j}-y'_{l_i}) \\
&\quad +(1 +\sqrt{\tfrac{1}{\A^2}+1})d_y(l_j,v) + \A d_x(l_j,v) \\
&= (1+\sqrt{\tfrac{1}{\A^2}+1})x'_{l_i}  + (x'_{l_j}-x'_{l_i}) + \A y'_{l_j} \\
&\quad +(1 +\sqrt{\tfrac{1}{\A^2}+1})d_y(l_j,v) + \A d_x(l_j,v).
\end{align*}
Using that $\A y'_{l_j}=\A x_{l_j}$, we obtain
\begin{align*}
d_t(u,v) &\le (1+\sqrt{\tfrac{1}{\A^2}+1})y_{l_j} + (1 +\sqrt{\tfrac{1}{\A^2}+1})d_y(l_j,v) + \A x
\\
&= (1+\sqrt{\tfrac{1}{\A^2}+1})y + \A x.
\end{align*}
completing the proof of Case 1.

\emph{Case 2}: There are vertices of $P$ inside $R(u, v)$. We distinguish two subcases, either $\A x\ge y$ or $\A x < y$.

\emph{Subcase $\A x\ge y$}: We split $R(u,v)$ into three regions formally defined as follows: $\mathcal{A} = \{ p \mid p \text{ is inside } R(u,v) \text{ such that } \A d_x(u,p) < d_y(u,p) \}$, $\mathcal{B} = \{ p \mid p \text{ is inside } R(u,v) \text{ such that }$ $\A d_x(u,p) \geq d_y(u,p) \text{ and } \A d_x(p,v) \geq d_y(p,v) \}$, $\mathcal{C} = \{ p \mid p \text{ is inside } R(u,v) \text{ such that }$ $\A d_x(p,v) < d_y(p,v) \}$. Informally, these three regions can be constructed by considering the line through $u$ and the line through $v$ parallel to the line through the diagonal of $R$ and labelling the resulting regions $\mathcal{A}$, $\mathcal{B}$, and $\mathcal{C}$ from left to right (see Figure~\ref{fig:threeRegions}(a)).

\begin{figure}[ht]
\centering
\includegraphics{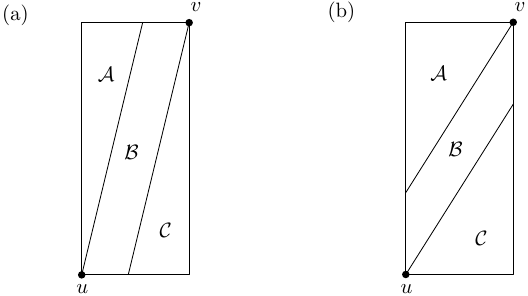}
\caption{(a) The three regions in $R(u, v)$ when $\A x \ge y$. (b) The three regions in $R(u, v)$ when $\A x < y$.}
\label{fig:threeRegions}
\end{figure} 

If there exists a vertex $p$ inside region $\mathcal{B}$, then we can apply induction on the pairs $(u,p)$, which satisfies $\A d_x(u,p) \ge d_y(u,p)$, and $(p,v)$, which satisfies $\A d_x(p,v) \ge d_y(p,v)$:
\begin{align*}
d_t(u,v) &\le d_t(u,p) + d_t(p,v)  \\
&\le (\A +\sqrt{\A^2+1})d_x(u,p) + d_y(u,p) + (\A +\sqrt{\A^2+1})d_x(p,v) + d_y(p,v) \\
&= (\A +\sqrt{\A^2+1})x + y.
\end{align*}

If there is no vertex inside region $\mathcal{B}$, we define $R_u$ to be the smallest scaled translate of $R$ that has $u$ on its lower left corner and some vertex $p \in \mathcal{A}$ in $R(u,v)$ on its boundary. Similarly, we define $R_v$ to be the smallest scaled translate of $R$ that has $v$ on its upper right corner and some vertex $q \in \mathcal{C}$ in $R(u,v)$ on its boundary. Since $R(u,v)$ is not empty, at least one of $p$ and $q$ must exist. Assume without loss of generality that $p$ exists. In this case we have that $\A d_x(p,v) > d_y(p,v)$ and the smallest homothet with $p$ and $v$ on its boundary is smaller than that of $u$ and $v$. If $(u,p)$ is an edge in the rectangle Delaunay triangulation, then we obtain that: 
\begin{align*}
d_t(u,v) &\le d_t(u,p) + d_t(p,v) \\
&= d_2(u,p) + d_t(p,v) \\
&\le d_x(u,p) + d_y(u,p) + (\A +\sqrt{\A^2+1})d_x(p,v) + d_y(p,v) \\
&\le (\A +\sqrt{\A^2+1})x + y.
\end{align*}
An analogous argument can be used if $q$ exists and $(v,q)$ is an edge in the rectangle Delaunay triangulation. 

Hence, it remains to consider the case where $(u,p)$ is not an edge, in which case $R_u$ is not empty. This implies that there exists a $p'\in \mathcal{C}$ such that $(u,p')$ is an edge. We have that $\A d_x(p',v) < d_y(p',v)$ and the smallest scaled translate of $R$ with $p'$ and $v$ on its boundary is smaller than that of $u$ and $v$. By the induction hypothesis, we have:
\begin{align*}
d_t(u,v) &\le d_t(u,p') + d_t(p',v) \\
&= d_2(u,p') + d_t(p',v) \\
&\le d_x(u,p') + d_y(u,p') + \A d_x(p',v) + \left( 1 +\sqrt{\tfrac{1}{\A^2}+1} \right)d_y(p',v) \\
&\le \A x + \left( 1 +\sqrt{\tfrac{1}{\A^2}+1} \right) y.
\end{align*}
Since $\A x \ge y$ and $\left( 1 +\sqrt{\tfrac{1}{\A^2}+1} \right)>1$, we have
\begin{align*}
d_t(u,v) &\le 1\A x + \left( 1 +\sqrt{\tfrac{1}{\A^2}+1} \right) y \\
&\le \left( 1 +\sqrt{\tfrac{1}{\A^2}+1} \right) \A x + 1y\\
&= \left( \A +\sqrt{\A^2+1} \right)x + y.
\end{align*}

\emph{Subcase $\A x< y$}: We split $R(u,v)$ into three regions formally defined as follows: $\mathcal{A} = \{ p \mid p \text{ is inside } R(u,v) \text{ such that } \A d_x(v,p) \ge d_y(v,p) \}$, $\mathcal{B} = \{ p \mid p \text{ is inside } R(u,v) \text{ such that }$ $\A d_x(v,p) < d_y(v,p) \text{ and } \A d_x(u,p) < d_y(u,p) \}$, $\mathcal{C} = \{ p \mid p \text{ is inside } R(u,v) \text{ such that }$ $\A d_x(u,p) \ge d_y(u,p) \}$. See Figure~\ref{fig:threeRegions}(b).

If there exists a vertex $p$ inside region $\mathcal{B}$, then we can apply induction on the pairs $(u,p)$, which satisfies $\A d_x(u,p) < d_y(u,p)$, and $(p,v)$, which satisfies $\A d_x(v,p) < d_y(v,p)$:
\begin{align*}
d_t(u,v) &\le d_t(u,p) + d_t(p,v)  \\
&\le \A d_x(u,p) + (1+\sqrt{\tfrac{1}{\A^2}+1})d_y(u,p) + \A d_x(p,v) + (1+\sqrt{\tfrac{1}{\A^2}+1})d_y(p,v) \\
&= \A x + (1+\sqrt{\tfrac{1}{\A^2}+1})y.
\end{align*}

If there is no vertex inside region $\mathcal{B}$, we define $R_u$ to be the smallest scaled translate of $R$ that has $u$ on its lower left corner and some vertex $p \in \mathcal{A}$ in $R(u,v)$ on its boundary. Similarly, we define $R_v$ to be the smallest scaled translate of $R$ that has $v$ on its upper right corner and some vertex $q \in \mathcal{C}$ in $R(u,v)$ on its boundary. Since $R(u,v)$ is not empty, at least one of $p$ and $q$ must exist. Assume without loss of generality that $p$ exists. In this case we have that $\A d_x(p,v) > d_y(p,v)$ and the smallest rectangle with $p$ and $v$ on its boundary is smaller than that of $u$ and $v$. If $(u,p)$ is an edge in the rectangle Delaunay triangulation, then we obtain that: 
\begin{align*}
d_t(u,v) &\le d_t(u,p) + d_t(p,v) \\
&= d_2(u,p) + d_t(p,v) \\
&\le d_x(u,p) + d_y(u,p) + (\A +\sqrt{\A^2+1})d_x(p,v) + d_y(p,v) \\
&\le (\A +\sqrt{\A^2+1})x + y.
\end{align*}
Since $\A x< y$, we have
\begin{align*}
d_t(u,v)&\le (\A +\sqrt{\A^2+1})x + y\\
&= (1+\sqrt{1+\tfrac{1}{\A^2}})\A x + 1y\\
&\le \A x + (1+\sqrt{\tfrac{1}{\A^2}+1})y.
\end{align*}
An analogous argument can be used if $q$ exists and $(v,q)$ is an edge in the rectangle Delaunay triangulation. 

Hence, it remains to consider the case where $(u,p)$ is not an edge, in which case $R_u$ is not empty. This implies that there exists a $p'\in \mathcal{C}$ such that $(u,p')$ is an edge. We have that $\A d_x(p',v) < d_y(p',v)$ and the smallest scaled translate of $R$ with $p'$ and $v$ on its boundary is smaller than that of $u$ and $v$. By the induction hypothesis, we have:
\begin{align*}
d_t(u,v) &\le d_t(u,p') + d_t(p',v) \\
&= d_2(u,p') + d_t(p',v) \\
&\le d_x(u,p') + d_y(u,p') + \A d_x(p',v) + \left( 1 +\sqrt{\tfrac{1}{\A^2}+1} \right)d_y(p',v) \\
&\le \A x + \left( 1 +\sqrt{\tfrac{1}{\A^2}+1} \right) y.
\end{align*}
This completes the proof of Case 2 and the theorem.
\end{proof}

We can now use Theorem~\ref{thm:upperBound1} to show an upper bound of the spanning ratio of the rectangle Delaunay triangulation. For any pair of vertices $u, v$ in the graph, if $\A d_x(u,v) \ge d_y(u,v)$ we have
\[
\frac{d_t(u,v)}{d_2(u,v)} < \frac{(\A +\sqrt{\A^2+1})x + y}{\sqrt{x^2+y^2}}.
\]
This function is maximized when $y/x = 1/(\A +\sqrt{\A^2+1})$, where the function is equal to \[\sqrt{2}\sqrt{1+\A^2+\A\sqrt{1+\A^2}}.\] 
On the other hand, when $\A d_x(u,v) < d_y(u,v)$, we can get 

\[
\frac{d_t(u,v)}{d_2(u,v)} < \frac{\A x + (1 +\sqrt{{\tfrac{1}{\A^2}}+1})y}{\sqrt{x^2+y^2}}.
\]
This function is maximized when $y/x = (1 +\sqrt{{\tfrac{1}{\A^2}}+1})/\A$, where the function value equals \[\sqrt{\A^2+2+2 \sqrt{1+{\tfrac{1}{\A^2}}}+\tfrac{1}{\A^2}},\]
which is at most $\sqrt{2}\sqrt{1+\A^2+\A\sqrt{1+\A^2}}$. This implies the main result of the paper.

\begin{theorem}
  The spanning ratio of the rectangle Delaunay triangulation is at most $\sqrt{2}\sqrt{1+\A^2+\A\sqrt{1+\A^2}}$, where $\A$ is the aspect ratio of the rectangle used in its construction. 
\end{theorem} 

Since it was already known that $\sqrt{2}\sqrt{1+\A^2+\A\sqrt{1+\A^2}}$ is a lower bound on the spanning ratio~\cite{BCR18}, we obtain that the bound of $\sqrt{2}\sqrt{1+\A^2+\A\sqrt{1+\A^2}}$ is tight. 

\section{Conclusion} 
\label{sec:conclusion}
We generalized and extended the proof technique of Bonichon~\etal~\cite{BGHP15} to prove the exact bound on the spanning ratio of rectangle Delaunay triangulation. While it has been known for quite some time that all generalized Delaunay graphs are plane spanners, a tight upper bound on the spanning ratio is known for only a small set of special cases: the equilateral triangle, the square and the regular hexagon. Our proof adds the class of all rectangles to this list by expressing the spanning ratio in terms of their aspect ratio. 

We note that while our proof is constructive, it is not constructive in a local sense, i.e., it doesn't immediately give rise to a local routing algorithm for the rectangle Delaunay triangulation. Future work therefore includes coming up with a routing algorithm that uses only local information (source, destination, and neighbours of the current vertex) to find a relative short path in rectangle Delaunay triangulation, as is known to exist for the Delaunay triangulation~\cite{bonichon_et_al:LIPIcs:2018:9485,BBCPR2017DelaunayJournal} and the equilateral triangle Delaunay triangulation~\cite{BFRV2015RoutingJournal}.

\newpage
\bibliography{references}
\end{document}